\documentclass[12pt,indentfirst]{article}

\usepackage{bbm}
\usepackage{psfrag}
\usepackage{tikz}
\usepackage{dsfont}
\usepackage{enumerate}
\usepackage{bm}
\usepackage[square]{natbib}  
\usepackage{amsmath}
\usepackage[utf8]{inputenc}
\usepackage[english]{babel}
\usepackage{amsthm}
\usepackage{amssymb}
\usepackage{algpseudocode,algorithm,algorithmicx}

\DeclareMathOperator*{\argmax}{arg\,max}

\newtheorem{theorem}{Theorem}
\newtheorem{corollary}{Corollary}
\newtheorem{lemma}{Lemma}
\newtheorem{definition}{Definition}

\newtheorem{remark}{Remark}

\def\cl#1{\begin{center}\Large #1\end{center}}
\def\ce#1{\centerline{#1}}

\algrenewcommand\algorithmicrequire{\textbf{Initialisation:}}
\algrenewcommand\algorithmicensure{\textbf{Output:}}

\begin{document}

\cl{Optimal Temperature Spacing for Regionally
Weight-preserving Tempering}

\medskip \ce{by}

\medskip \ce{Nicholas G. Tawn and Gareth O.\ Roberts
}

\begin{abstract}

Parallel tempering is popular method for allowing MCMC algorithms to
properly explore a $d$-dimensional multimodal target density.  One
problem with traditional  power-based parallel tempering for multimodal targets is that the
proportion of probability mass associated to the modes can change for different inverse-temperature
values, sometimes dramatically so. Complementary work by the authors proposes a novel solution involving auxiliary targets  that preserve regional weight upon powering up the density. This paper attempts to address the question of how to choose the temperature spacings in an optimal way when using this type of weight-preserving approach.
The problem is analysed in a tractable setting for computation of the expected squared jumping distance which can then be optimised with regards to a tuning parameter.
The conclusion is that for an appropriately constructed regionally weight-preserved tempering algorithm targeting a $d$-dimensional target distribution, the consecutive temperature spacings should behave as $\mathcal{O}\left(d^{-1/2}\right)$ and this induces an optimal acceptance rate for temperature swap moves that lies in the interval $[0,0.234]$.
\end{abstract}

\section{Introduction}

Consider the problem of simulating from a $d$-dimensional target density, $\pi(x)$,  with respect to Lebesgue measure on $\mathbb{R}^d$, which is only known up to a scaling constant. A popular approach is to use Markov chain Monte Carlo (MCMC) which uses a Markov chain that is constructed in such a way that the invariant distribution is $\pi$.

If $\pi$ exhibits multimodality then the majority of MCMC algorithms which use tuned localised proposal mechanisms e.g.,\ \cite{roberts1997weak}, \cite{roberts2001optimal} and \cite{Mangoubi2018}, can often fail to explore the state space leading to bias samples. A practical approach to overcoming the multimodality issue is the  \textit{Parallel Tempering} (PT) algorithm, which augments the state space with auxiliary target distributions that enable the chain to rapidly traverse the entire state space.

The major problem with the auxiliary targets that are typically chosen for the PT algorithm is that in general they do not preserve regional mass, see \cite{woodard2009conditions}, \cite{woodard2009sufficient} and \cite{Bhatnagar2016}. This problem can result in the necessary run-time of the  PT algorithm increases exponentially with the dimensionality of the problem.

Novel methodology presented in \cite{GarethJeffNick} and  \cite{NTawnThesis} attempts to mitigate the modal weight degeneracy in certain toy situations. In particular  \cite{GarethJeffNick} introduces and analyses the performance of the new prototype Hessian adjusted tempering algorithm (HAT). The HAT algorithm is similar to the PT approach but partitions the state space into regions and upon power-tempering the target, the regions are re-weighted to (approximately) preserve their mass at all temperatures.

 \cite{atchade2011towards} and \cite{roberts2014minimising} show that the temperature spacing for the PT algorithm needs to scale as $\mathcal{O}(d^{-1/2})$ where $d$ denotes the dimensionality of the problem. Furthermore they demonstrate the existence of a limiting efficiency curve as a function of the  temperature swap acceptance rate that is optimised at 0.234.

This paper analyses the optimal spacings for a general \textit{regionally-weight preserved tempering algorithm}. The conclusion is similar to that of \cite{atchade2011towards}. Temperature spacings need to be scaled as $\mathcal{O}(d^{-1/2})$ and there exists a limiting efficiency curve that is a function of the  temperature swap acceptance rate that is optimised at a value in the interval [0,0.234].

The layout of the paper is as follows. Section~\ref{subsec:parallel} introduces the standard PT algorithm and then defines a regionally weight-preserved tempering approach. Section~\ref{Sec:OptscalRegional} presents and proves the two main theoretical contributions of the paper:  Theorem~\ref{Thr:optscal}, which introduces the main optimal spacing result using an extension of an elegant argument from \cite{Sherlock2006} to establish an optimal acceptance rate range, and Corollary~\ref{Cor:corrgeom}, which establishes a sufficient condition on the target density that leads to a geometric tempering schedule making algorithm setup/tuning significantly simpler. A brief conclusion is given in Section~\ref{conc}.

\section{Regionally Weight-Preserving PT Algorithm}
\label{subsec:parallel}

Many of the methods employed  to overcome the issues of multimodality in MCMC  use state space augmentation e.g.\ \cite{Wang1990a}, \cite{geyer1991markov}, \cite{marinari1992simulated}, \cite{Neal1996}, \cite{kou2006discussion}, \cite{2017arXiv170805239N}. Typically they utilise auxiliary distributions that can be more readily explored by suitable Markov chains and the resulting mixing information is then passed on to aid  inter-modal mixing in the desired target. A  convenient approach of this type uses power-tempered target distributions i.e.\ the target distribution at inverse temperature level $\beta$, for $\beta \in (0,1]$, is defined as \begin{equation}
\pi^\beta(x)\propto \left[\pi(x)\right]^\beta \label{powtar}
\end{equation} 
Such targets are the most common choice of auxiliary target when augmenting the state space for use in the popular simulated tempering (ST) and parallel tempering (PT) algorithms introduced in \cite{marinari1992simulated} and \cite{geyer1991markov}. For each algorithm one needs to choose a sequence of $n+1$ ``inverse temperatures'', $\Delta=\{\beta_0,\ldots,\beta_n\}$, where $0 \leq \beta_n<\beta_{n-1}<\ldots <\beta_1<\beta_0=1$ with the specification that a Markov chain sampling from the target distribution $\pi^{\beta_n}(x)$ can mix well across the entire state space.

 The PT algorithm runs a Markov chain on the augmented state space, $\mathcal{X}^n$, targeting an invariant distribution given by
\begin{eqnarray}
\pi(x_0,x_1,\ldots,x_n) \propto \pi^{\beta_0}(x_0)\pi^{\beta_1}(x_1)\ldots\pi^{\beta_n}(x_n).\nonumber
\end{eqnarray}
From an initialisation point for the chain the PT algorithm  alternates between two types of Markovian move
\begin{itemize}
		\item \textit{Temperature swap} moves that propose to swap the chain locations between a pair of adjacent temperature components. It is these swap moves that will allow mixing information from the hot, rapidly-mixing temperature level to be passed to aid mixing at the cold target state. \item  \textit{Within temperature} Markov chain moves that use standard localised MCMC schemes to update each of the $x_i$ whilst preserving marginal invariance. Note that whilst the choice of the \textit{within temperature} Markov chain moves is certainly important in accelerating the performance of the algorithm e.g.\ \cite{Ge2017} and \cite{Mangoubi2018}, the focus of this paper is entirely on the \textit{temperature swap} transitions.
\end{itemize}

To perform the \textit{temperature swap} move a pair of temperatures is chosen uniformly from the set of all adjacent pairs. Denote the marginal components $x_i$ and $x_{i+1}$ at inverse temperatures $\beta_{i}$ and $\beta_{i+1}$ respectively. To preserve detailed balance and therefore invariance to $\pi_n(\cdot)$ then the swap move is accepted with probability
\begin{equation}
A=\mbox{min}\Bigg( 1,\frac{\pi^{\beta_{i+1}}(x_i)\pi^{\beta_{i}}(x_{i+1})}{\pi^{\beta_{i}}(x_i)\pi^{\beta_{i+1}}(x_{i+1})} \Bigg).
\label{eq:parstd1}
\end{equation}

Note that this is the standard form of the PT \textit{temperature swap} move procedure and this will be the focus of this paper; however it is worth noting that recent work in \cite{Tawn2018} shows that in certain circumstances this can be modified to give accelerated mixing through the temperature schedule.

The regionally weight-preserved parallel tempering (RWPPT) algorithm is essentially the same as the standard PT algorithm but uses modified auxiliary target distributions. These auxiliary targets simply preserve mass on some chosen partition components of the state space across all temperatures in the schedule; Definition~\ref{eq:defRWPTT} makes this notion concrete.
\begin{definition}[Regionally Weight-Preserved Tempered Target (RWPTT)] \label{eq:defRWPTT}
For some $K \in \mathbb{N}$, consider a partition of the $d$-dimensional state space $\mathcal{X}$ such that $\mathcal{X}= \bigsqcup_{k=1}^K A_k$ (where $\bigsqcup$ denotes the disjoint union). Denoting \[ w_k = \int_{A_k} \pi(x) dx\] then the RWPTT at inverse temperature level $\beta$ is given by
\begin{equation}
	\pi_\beta(x) \propto  \sum_{k=1}^K  w_k  \frac{\pi^\beta(x)}{\int_{A_k}\pi^\beta (x) dx} \mathds{1} _{[x \in A_{k}]}. \label{eq:regweightpretarg}
\end{equation}
\end{definition}
It is these RWPTTs are to be utilised in the PT framework (such an approach is taken in \cite{GarethJeffNick}) and the resulting algorithm is defined as follows:

\begin{definition}[Regionally Weight-Preserved Parallel Tempering (RWPPT) Algorithm]
The RWPPT algorithm is defined to be a PT algorithm which targets the invariant distribution 
\begin{eqnarray}
\pi_n(x_0,x_1,\ldots,x_n) \propto \pi_{\beta_0}(x_0)\pi_{\beta_1}(x_1)\ldots\pi_{\beta_n}(x_n).\nonumber
\end{eqnarray} 
where $\pi_\beta(\cdot)$ are the RWPTTs defined in Definition~\ref{eq:defRWPTT}. The algorithm proceeds identically to the PT algorithm combining within temperature and temperature swap moves. Note that a temperature swap proposal between adjacent temperature levels $\beta_i$ and $\beta_{i+1}$ is now accepted with probability
\begin{equation}
A=\mbox{min}\Bigg( 1,\frac{\pi_{\beta_{i+1}}(x_i)\pi_{\beta_{i}}(x_{i+1})}{\pi_{\beta_{i}}(x_i)\pi_{\beta_{i+1}}(x_{i+1})} \Bigg).
\nonumber
\end{equation}
\end{definition}

\section{Optimal spacing of a RWPPT Algorithm}
\label{Sec:OptscalRegional}

The aim is to maximise the efficiency of the inter-modal exploration of the RWPPT algorithm. This can be done by minimising the expected time taken to transfer the (assumed rapidly mixing) hot state mixing information, at temperature level $\beta_n$, to  the cold target state, at temperature level $\beta_0$. A measure of the  efficiency of mixing through the temperature schedule is the Expected Squared Jumping Distance ($ESJD_\beta$), see \cite{atchade2011towards}. Indeed, $ESJD_\beta$ becomes the natural measure of efficiency if there is an associated limiting diffusion process, \cite{roberts2014minimising}. For the RWPPT approach studied here, $ESJD_\beta$ is likely to be the natural choice since \cite{GarethJeffNick} derives a limiting skew-Brownian motion result for the special case of a Gaussian mixture target distribution that is analogous to the RWPPT setup.

\begin{definition}[$ESJD_\beta$]
Assume the framework of a RWPPT algorithm and suppose that a  temperature swap move between inverse temperature levels $\beta$ and  $\beta^{'}= \beta+ \epsilon$, for some $\epsilon >0$, has been proposed. Then the $ESJD_\beta$ is defined as
\begin{eqnarray}
	ESJD_\beta &=& \mathbb{E}\left[ (\gamma -\beta) ^2 \right] \label{eq:ESJDw}
\end{eqnarray}
where $\gamma$ is the random variable taking the values $\beta$ if the proposed temperature swap move is rejected or $\beta^{'}$, if the move is accepted. The expectation is with respect to the target distribution, $\pi_n(\cdot)$, since it assumed that the Markov chain is in  stationarity.
\end{definition}

In order to mix effectively through the temperature schedule one needs a strategy that balances making overly ambitious large jump proposals which have low acceptance probabilities, against under ambitious small jump proposals with high acceptance probabilities; both of which lead to slow mixing through the temperature schedule. By tuning consecutive temperature spacings to maximise the $ESJD_\beta$ between levels then a strategy balancing ambition and non degenerate acceptance should be reached.

Section~\ref{sec:ass} provides the assumptions and toy setting framework necessary for a tractable analysis of this optimisation problem for a RWPPT algorithm; concluding with the key theoretical contribution of this paper, Theorem~\ref{Thr:optscal}.

\subsection{Optimal Temperature Spacing} \label{sec:ass}

Assume a  $d$-dimensional state space $\mathcal{X}_d \subseteq \mathbb{R}^d$,  is constructed from a disjoint union of hypercubes i.e.\
\begin{equation}
	\mathcal{X}_d = \bigsqcup_{j=1}^K A_{(j,d)}\label{eq:disjunion}
\end{equation}
with 
\begin{equation}
	 A_{(k,d)}= A_k^1 \otimes \ldots \otimes A_k^d = [a_k^1,b_k^1] \otimes \ldots \otimes [a_k^d,b_k^d]  \label{eq:hyp}
\end{equation}
so that for $i ,j \in \{ 1,\ldots, d \}$, $(b_k^i-a_k^i)=(b_k^j-a_k^j)$. It is these hypercubes that will be the regions upon which mass is preserved for the  RWPTTs.

For tractability it will be assumed that the target distribution $\pi(\cdot)$ has conditionally independent marginals within each hypercube:  \begin{equation} \pi_j^{\beta}(x) =\left[\prod_{i=1}^d \frac{ f_{(k,i)}^{\beta}(x_i)}{\int_{A_k^i}f_{(k,i)}^{\beta}(z)dz}\right] \mathds{1} _{[\mathbf{x} \in A_{(k,d)}]}. \label{eq:margiid} \end{equation}
Then the RWPTT for any $\beta \in (0,\infty)$ is given by
\begin{equation}  \pi_{\beta}(x) = \sum_{k=1}^{K} w_{k}\pi_k^{\beta}(x) . \label{eq:WPTTpf}\end{equation}

It is assumed that the marginal components $f_{(k,i)}^{\beta}(\cdot)$ have a shifted iid form on the corresponding region $A_k$. That is, for each $k\in \{ 1,\ldots, K \}$ there is a density, denoted $f_k(\cdot)$, such that for all $i  \in \{ 1,\ldots, d \}$
\begin{equation}
	f_{(k,i)}(x_i) = f_k(x_i-\mu_k^i) \mathbbm{1}_{\left[x_i \in A_k^i\right]}, \label{eq:iidshift}
\end{equation}
where $\mu_k^i= \frac{a_k^i+b_k^i}{2}$. Furthermore, it is assumed that for all $k\in \{ 1,\ldots, K \}$ $f_k(\cdot) \in C^3$ and that for $a \in \{1,2\}$ and $b \in \{0,1\}$ there exists integrable functions $g_{(a,b,k)}(x): \mathbb{R} \rightarrow \mathbb{R}$ such that 
\begin{equation}
	\left| \log f_k(x) \right|^b\left|\frac{\partial^a}{\partial \beta^a} f_k^\beta (x) \right| \le g_{(a,b,k)}(x). \label{eq:techcond1}
\end{equation}

As in \cite{atchade2011towards}, the aim  is to maximise the the $ESJD_\beta$ in \eqref{eq:ESJDw}  but now under the setting established above with the state space given in \eqref{eq:hyp} and the corresponding RWPTTs constructed by \eqref{eq:margiid} and \eqref{eq:WPTTpf}.  Further to this we aim to associate this optimal spacing with a corresponding limiting optimal temperature swap acceptance rate, which is defined as follows:
\begin{definition}[Temperature Swap Acceptance Rate]
Assume the framework of a RWPPT algorithm for a $d$-dimensional target and take two consecutive inverse temperature levels $\beta$ and  $\beta^{'}= \beta+ \epsilon$, for some $\epsilon = \ell/d^{1/2}$. Then the acceptance rate for the temperature swap moves between these two levels is defined as
\begin{eqnarray}
	a(\ell,d) = \mathbb{E}_{\pi_n} \left[ \min \left(1 ,\frac{\pi_{\beta^{'}}(x)\pi_{\beta}(y)}{\pi_{\beta^{'}}(y)\pi_{\beta}(x)} \right)\right] \label{def:dnonlimacc}
\end{eqnarray}
with the corresponding limiting temperature swap acceptance rate defined as
\begin{equation}
	a(\ell) = \lim_{d \rightarrow \infty} a(\ell,d). \label{def:limacc}
\end{equation}
\end{definition}

Denoting the CDF of a standard Gaussian distribution by $\Phi (\cdot)$, then the main result is as follows:

\begin{theorem}[Optimal Spacing for a RWPPT Algorithm]
Consider applying the RWPPT algorithm under the setting of \eqref{eq:disjunion}, \eqref{eq:hyp}, \eqref{eq:margiid}, \eqref{eq:WPTTpf}, \eqref{eq:iidshift} and \eqref{eq:techcond1}. Consider two consecutive inverse temperature levels $\beta$ and $\beta^{'}=\beta +\epsilon$ where  $\epsilon =\ell/d^{1/2}$. Then the following holds:
\begin{enumerate}[i]
	\item The $ESJD_\beta$, from \eqref{eq:ESJDw}, satisfies  \begin{equation}
E(\ell):= \lim_{d\rightarrow \infty}d (ESJD_\beta) = \sum_{j=1}^K\sum_{m=1}^K  2 w_j w_m  \ell^2 \Phi\left(-\frac{\ell\sigma_{j,m}(\beta)}{2}\right)  \label{eq:theoremstatement}
\end{equation}
where $\sigma_{j,m}(\beta)=  \left( I_j(\beta)+I_m(\beta)\right)^{\frac{1}{2}}$ and $I_k(\beta)=\mbox{Var}_{f_{k}^{\beta}}\left[ \log f_{k}(X) \right] $.
\item Denoting $\hat{\ell}:= \argmax_\ell{E(\ell)}$, then the corresponding limiting temperature swap acceptance rate, from \eqref{def:limacc}, is given by \[ a(\hat{\ell}) \in [0,0.234]. \]
\end{enumerate}
\label{Thr:optscal}
\end{theorem}

\begin{remark} The assumption that $\mathcal{X}_d$ is constructed from a set of hypercubes can be weakened to hyper-rectangles by rescaling the marginals. So for each $k\in \{1,\ldots,K \}$ the hyper-rectangle given by \[A_{(k,d)}= A_k^1 \otimes \ldots \otimes A_k^d = [a_k^1,b_k^1] \otimes \ldots \otimes [a_k^d,b_k^d]  \]
omits the marginal target defined on this support given by
\[
	f_{(k,i)}(x) = \frac{1}{\sigma_{(k,i)}}f_k\left(\frac{x-\mu_k^i}{\sigma_{(k,i)}}\right),
\]
where 
\[  \left(\frac{b_k^i-a_k^i}{\sigma_{(k,i)}}\right)= C ~~~~~ \text{for some} ~~~C\in \mathbb{R}.\] 
\end{remark}
\begin{remark} Note that the assumption that $\mathcal{X}_d \subseteq \mathbb{R}^d$ can be weakened to any state space where $\mathcal{X}_d = \mathcal{X} \otimes \ldots \otimes \mathcal{X}$ where $\mathcal{X}$ is any ordered set.
\end{remark}
\begin{proof}

The proof of Theorem~\ref{Thr:optscal} comprises 3 key steps that are presented as lemmata below. Lemma~\ref{lemma:asymWB} establishes that the logged swap move acceptance ratio converges weakly to a Gaussian mixture distribution as $d\rightarrow \infty$. Lemma~\ref{asymlemmaesjd} exploits the weak convergence to derive a limiting formula for $ESJD_\beta$ that is a function of the tuning parameter $\ell$ and can thus be maximised to give an optimal value, $\hat{\ell}$, . Finally, Lemma~\ref{accstepw}  shows that the acceptance rate induced by choosing $\ell = \hat{\ell}$ lies in the interval [0,0.234].

For notational clarity, suppose that a proposed temperature swap move between the temperature marginals $x$ and $y$ at the consecutive inverse temperatures level $\beta$ and $\beta^{'}=\beta+\epsilon$. The aim is to maximise  the $ESJD_\beta$ which can be immediately decomposed as
\begin{eqnarray}
\mathbb{E}_{\pi_n}\left[  (\gamma-\beta)^2\right] &=& \epsilon^2 \times a(\ell,d) \nonumber\\
&=& \epsilon^2 \times  \mathbb{E}_{\pi_n}\left[ \min(1 , e^B) \right], \label{eq:ESJDregional}
\end{eqnarray}
where $a(\ell,d)$ is defined in \eqref{def:dnonlimacc} and
\begin{eqnarray}
B &=&\left[ \log \left( \pi_{\beta^{'}}(x)) \right) -\log \left( \pi_{\beta}(x)) \right)\right] -\left[ \log \left( \pi_{\beta^{'}}(y)) \right)  - \log \left( \pi_{\beta}(y)) \right) \right] \nonumber \\
&=:& H_d^{\beta}(x) -H_d^{\beta}(y).
\label{eq:Bw}
\end{eqnarray}

\begin{lemma}[Asymptotic Gaussianity of $B$] \label{lemma:asymWB}
Under the setting of Theorem~\ref{Thr:optscal} and with $B$ defined in \eqref{eq:Bw}; as $d \rightarrow \infty$,  $B$  converges weakly to a Gaussian mixture distribution such that
\begin{equation}
	B \Rightarrow \sum_{j=1}^K\sum_{m=1}^K w_j w_m N\left( -\frac{\ell ^2}{2}\left( I_j(\beta)+I_m(\beta)\right),  \ell ^2 \left( I_j(\beta)+I_m(\beta)\right) \right),
\end{equation}
where $I_k(\beta)=\mbox{Var}_{f_{k}^{\beta}}\left[ \log f_{k}(X) \right] $.
\end{lemma}
\begin{proof}
See the Appendix in Section~\ref{app:pf1}.
\end{proof}

Since the function of $B$ given by $\min \left\{ 1, e^B \right\}$ is bounded then the weak convergence of $B$ to Gaussianity from Lemma~\ref{lemma:asymWB}, can be exploited in the following lemma.

\begin{lemma}[Asymptotic form of $ESJD_\beta$] \label{asymlemmaesjd}
Under the setting of Theorem~\ref{Thr:optscal}  then
\begin{equation}
E(\ell) :=\lim_{d \rightarrow \infty} (d \times ESJD_\beta) = \ell^2\sum_{j=1}^K\sum_{m=1}^K w_j w_m \left(   2 \Phi\left(-\frac{\ell\sigma_{j,m}(\beta)}{2}\right)  \right), \label{eq:ESJDfinal}
\end{equation}
where
$\sigma_{j,m}(\beta)= \left( I_j(\beta)+I_m(\beta)\right)^{\frac{1}{2}}$.
\end{lemma}
\begin{proof}
See the Appendix in Section~\ref{app:pf2}.
\end{proof}

To find the limiting optimal spacing,  $E(\ell)$ in (\ref{eq:ESJDfinal}) is maximised with respect to $\ell$, giving a value denoted $\hat{\ell}$. 

\begin{lemma}[Optimal Limiting Acceptance Rate] \label{accstepw}
The optimising value, $\hat{\ell}$, established in Lemma~\ref{asymlemmaesjd}, corresponds to a consecutive temperature level spacing with a temperature swap acceptance rate, $\hat{a}$ in the range of $[0,0.234]$.
\end{lemma}
\begin{proof}
See the Appendix in Section~\ref{app:pf3}.
\end{proof}

\end{proof}

\subsection{Geometric Temperature Spacings}

A geometric temperature schedule is induced when one chooses $\ell \propto \beta$ and so for some constant $C\in (1,\infty)$, $ \beta^{'} = C \beta$. The power of this is that the choice of the one-dimensional parameter $C$ determines the whole temperature schedule with $n-1$ spacings. In general this approach will not be optimal since  $\hat{\ell}$ from Theorem~\ref{Thr:optscal} can be  non-linearly dependent on $\beta$. However, for certain target distributions the dependence of $\hat{\ell}$ on $\beta$ is linear. The following corollary establishes a sufficient condition on the target for this to be the case.

\begin{corollary}
\label{Cor:corrgeom} 
Take a target distribution of the form specified in Theorem~\ref{Thr:optscal}, i.e.\ given by \eqref{eq:margiid} and \eqref{eq:WPTTpf}, but now where the marginal components satisfy
\begin{eqnarray}
	f_{(k,i)}(x_i) =f_k(x_i-\mu_k^i) \mathbbm{1}_{\left[x_i \in A_k^i\right]} \propto \exp \left(\left| \frac{x_i-\mu_k^i}{\sigma_k}\right|^z\right) \mathbbm{1}_{\left[x_i \in A_k^i\right]} \label{eq:expfam}
\end{eqnarray}
for some $z \in \mathbb{N} $. Then in the limit as $d \rightarrow \infty$ the $ESJD_\beta$ is maximised when $\ell$ is chosen to maximise \eqref{eq:theoremstatement} in Theorem~\ref{Thr:optscal}; but now with this choice of $\ell$ inducing an optimal acceptance rate $\hat{a}=0.234$. Furthermore, 
$\hat{\ell} \propto \beta$ and so the optimal setup is geometrically spaced. 

\end{corollary}
\begin{proof}
 
Recalling the result in \eqref{eq:theoremstatement} of Theorem~\ref{Thr:optscal} with $\sigma_{j,m}(\beta)=  \left( I_j(\beta)+I_m(\beta)\right)^{\frac{1}{2}}$ where $I_k(\beta) =  \mbox{Var}_{f_{k}^{\beta}}\left[ \log f_{k}(X) \right] $. Computing $I_k(\beta)$ when $f_k$ take the form given in \eqref{eq:expfam}:
\begin{eqnarray}
\mbox{Var}_{f_{k}^{\beta}}\left[ \log f_{k}(X) \right] 
=\mbox{Var}_{f_{k}^{\beta}}\left(\left| \frac{x_i-\mu_k^i}{\sigma_k}\right|^z\right) = \frac{1}{\beta^2} \mbox{Var}_{f_{k}^{\beta}}\left(\beta \left| \frac{x_i-\mu_k^i}{\sigma_k}\right|^z\right) \nonumber
\end{eqnarray}
and a simple change of variables for the final variance term on the RHS shows that 
\[\mbox{Var}_{f_{k}^{\beta}}\left(\beta\left| \frac{x_i-\mu_k^i}{\sigma_k}\right|^z\right) = D \]
for some constant $D \in \mathbb{R}_{+}$ that does not depend on $\beta$ or $k$. Consequently, $I_k(\beta) = D/ \beta^2 $ for all $k = 1,\ldots, K$ and so for all pairs $j$ and $m$ then \[\sigma_{j,m}(\beta)= \frac{\sqrt{2D}}{ \beta}.\]
Consequently, the limiting form of the $ESJD_{\beta}$ in \eqref{eq:theoremstatement} simplifies to
\begin{equation}
2 \ell^2\Phi\left(-\frac{\ell}{\beta\sqrt{2}}\right)   \nonumber
\end{equation}
and a simple substitution of $u=\ell/\beta$ shows that the maximising  $\ell$ is such that $\hat{\ell}\propto \beta$ and this induces an optimal acceptance rate of
\begin{equation}
2\Phi \left(-\frac{\hat{\ell}}{\beta\sqrt{2}}\right)  =0.234 ~~~~~(3 s.f.).
	\nonumber
\end{equation}
\end{proof}

\section{Conclusion and Further work} \label{conc}
The lack of regional weight-preservation upon power tempering can lead to exponentially poor performance in high dimensions. Certain prototype algorithms to overcome this problem can be considered to be attempting to stabilise the regional mass similarly to the RWPTTs introduced in Definition~\ref{eq:defRWPTT}. Setting up a tempering schedule to efficiently pass the mixing information through the schedule is important to attain good performance. This paper has made two theoretical contributions to analyse exactly this issue, Theorem~\ref{Thr:optscal} and Corollary~\ref{Cor:corrgeom}.

Theorem~\ref{Thr:optscal} establishes that the scalability of the mixing in the temperature space with dimension requires consecutive inverse temperature spacings that are $\mathcal{O}(d^{-1/2})$; so with the RWM type moves proposing swaps between consecutive temperature levels then one expects the mixing time in the temperature space to be $\mathcal{O}(d)$. In addition to this, the theorem gives guidance on optimal setup to a practitioner who may be implementing a strategy that uses a RWPPT algorithm since it suggests that the acceptance ratio should be $\le 0.234$.

Corollary~\ref{Cor:corrgeom} extends the theorem when more assumptions are made regarding the regional marginal densities. The optimal swap acceptance rate becomes exactly $0.234$ and this induces a geometric temperature schedule making tuning for an optimal setup significantly simpler.

In conclusion, the results here are complementary to recent work in \cite{GarethJeffNick}, \cite{Tawn2018} and \cite{NTawnThesis}. These methods are still prototype but the work in this paper enhances the understanding of the scalability these new approaches. Also it gives theoretical insight that can be used to direct the future directions of such approaches when the modes/regions exhibit inhomogeneous behaviour.

The authors are currently combining the ideas and theoretical contributions from this work and that of the companion papers \cite{GarethJeffNick}, and \cite{NTawnThesis} to develop and analyse the scalability of a novel practically-applicable algorithm for exploration of multimodal target distributions. The results presented here are proving crucial to informing optimisation and the computational complexity of the proposed algorithm.

\section{Appendix}
\label{sec:appendix}

\subsection{Proof of Lemma~\ref{lemma:asymWB}} \label{app:pf1}
\begin{proof}
Due to the shifted iid form of the marginal targets from \eqref{eq:iidshift}, make the change of variables $x_i^s=x_i-\mu_k^i$ when $x \in A_{k}$. For notational convenience herein, denote \[  
C_{k}(\beta):=\int_{A_k^i}f_{(k,i)}^{\beta}(z)dz
\] where it should be noted that the LHS does not depend on $i$ due to the shifted iid form of the marginal components.

Consider the term $H_d^{\beta}(x)$ from (\ref{eq:Bw}), and use a Taylor expansion to third order:
\begin{eqnarray}
&&\sum_{k=1}^{K} \left[\sum_{i=1}^d  \log \left( \frac{f_{k}^{\beta^{'}}(x_i^s)}{C_{k}(\beta^{'})}\right) - \log \left( \frac{f_{k}^{\beta}(x_i^s)}{C_{k}(\beta) } \right)\right] \mathds{1} _{[x \in A_{k}]}\nonumber\\
	&=&\sum_{k=1}^{K} \left[\sum_{i=1}^d \epsilon \log\left( f_{k}(x_i^s) \right)-\log [C_{k}(\beta^{'})]  +\log [ C_{k}(\beta)] \right] \mathds{1} _{[x \in A_{k}]} \nonumber\\
	&=& \sum_{k=1}^{K} \Bigg[\sum_{i=1}^d \epsilon \log\left( f_{k}(x_i^s) \right)-\epsilon \frac{\partial}{\partial \beta}\log \left(C_{k}(\beta)\right)\nonumber\\
	&& -\frac{\epsilon ^2}{2} \frac{\partial^2}{\partial^2 \beta}\log \left( C_{k}(\beta)\right)- \frac{\epsilon^3}{6} \frac{\partial^3}{\partial^3 \beta}\log \left( C_{k}(\beta+\xi_{k}) \right) \Bigg] \mathds{1} _{[x \in A_{k}]} 
	\label{eq:TayHx}
\end{eqnarray}
where  $0<|\xi_{k}|<\epsilon$ is the mean value Taylor remainder. Analysing the first order derivative term in (\ref{eq:TayHx}) and utilising the assumption in \eqref{eq:techcond1} that permits the use of the Leibniz integral rule
\begin{eqnarray}
 M_{k}(\beta):=\frac{\partial}{\partial \beta}\log \left( C_{k}(\beta) \right) &=& \frac{ \int_{A_k^i}\log\left(f_{(k,i)}(x)\right)f_{(k,i)}^{\beta}(x)dx}{\int_{A_k^i}f_{(k,i)}^{\beta}(z)dz}\nonumber \\
&=& \mathbb{E}_{f_{(k,i)}^{\beta}}\left[ \log f_{(k,i)}\left(X\right) \right]\nonumber \\
&=&  \mathbb{E}_{f_{k}^{\beta}}\left[ \log (f_{k}\left(X\right) \right]. \label{eq:Mbetaintro}
\end{eqnarray}

Now consider the second order derivative term in (\ref{eq:TayHx}); a similar calculation to \eqref{eq:Mbetaintro}, that utilises assumption \eqref{eq:techcond1}, shows that
\begin{eqnarray}
 I_{k}(\beta):=\frac{\partial^2}{\partial^2 \beta}\log \left(C_{k}(\beta) \right) 
= \mbox{Var}_{f_{k}^{\beta}}\left[ \log f_{k}(X) \right] . \nonumber
\end{eqnarray}

For notational convenience, and due to the ``shifted'' iid setup, herein the following notation is used: \[ J_k(\beta):=\frac{\partial^3}{\partial^3 \beta}\log \left( \int_{A_k^i}f_{(k,i)}^{\beta}(z)dz\right). \]
Using this notation one can rewrite \eqref{eq:TayHx} as
\begin{eqnarray}
	H_d^{\beta}(x) &=&\sum_{k=1}^{K} \Bigg[\sum_{i=1}^d \epsilon \log\left( f_{k}(x_i^s) \right)-\epsilon M_k(\beta) \nonumber \\ &&-\frac{\epsilon ^2}{2}I_k(\beta)- \frac{\epsilon^3}{6} J_k(\beta+\xi_k)  \Bigg] \mathds{1} _{[x\in A_{k}]}. \label{Hxbettay}
\end{eqnarray}
By identical methodology with $y_i^s=y_i-\mu_k^i$ when $y \in A_k^i$ then
\begin{eqnarray}
	H_d^{\beta}(y) &=&-\sum_{k=1}^{K} \Bigg[\sum_{i=1}^d \epsilon M_k(\beta) +\frac{\epsilon ^2}{2} I_k(\beta)\nonumber \\  &&+ \frac{\epsilon^3}{6} J_k(\beta+\xi_k)-\epsilon \log\left( f_{k}(y_i^s) \right)  \Bigg] \mathds{1} _{[y \in A_{k}]}. \label{eqHtayy}
\end{eqnarray}
Using Taylor expansion, there exists $0<|\xi_{T_k}|<\epsilon$ such that 
\begin{equation}
 M_k(\beta')= M_k(\beta)+ \epsilon I_k(\beta)+ \frac{\epsilon^2}{2} J_k(\beta+\xi_{T_k}). \label{eq:tayMbet}
\end{equation}
Substituting $M_k(\beta)$ from \eqref{eq:tayMbet} into \eqref{eqHtayy}
\begin{eqnarray}
	H_d^{\beta}(y) &=&-\sum_{k=1}^{K} \Bigg[\sum_{i=1}^d \epsilon M_k(\beta^{'}) -\frac{\epsilon ^2}{2} I_k(\beta)-\epsilon \log\left( f_{k}(y_i^s) \right)\nonumber \\
	&&~~~+ \frac{\epsilon^3}{6} \left(J_k(\beta+\xi_k)- 3J_k(\beta+\xi_{T_k}) \right) \Bigg] \mathds{1} _{[y \in A_{k}]}. \label{eq:hytayfin}
\end{eqnarray}
Combining \eqref{Hxbettay} and \eqref{eq:hytayfin} then $B$ can be written as
\begin{eqnarray}
	B &=&  \sum_{j=1}^K\sum_{m=1}^K \sum_{i=1}^d \Bigg[r_{(x,y,j,m),i}^{(\beta,\beta^{'})} - \frac{\epsilon^3}{2}  J_k(\beta+\xi_{T_k}) \Bigg]\mathds{1} _{[x \in A_{j}]}\mathds{1} _{[y \in A_{m}]}.  \nonumber
\end{eqnarray}
where
\begin{eqnarray}
	r_{(x,y,j,m),i}^{(\beta,\beta^{'})} &=& \Bigg[ \epsilon  [\log\left( f_{j}(x_i^s)\right)-M_j(\beta) ]-\epsilon [\log\left( f_{m}(y_i^s) \right)-M_m(\beta^{'})]\nonumber \\
	&&~~-\frac{\epsilon ^2}{2}\left( I_j(\beta)+I_m(\beta)\right)\Bigg] 
\end{eqnarray}

Let $E_k^x= \{x \in A_{k}\}$,  then using the fact that the $r_{(x,y,l,m),i}^{(\beta,\beta^{'})}$ are independent and identically distributed for all $i \in \{ 1, \ldots,d\}$
\begin{equation}
	\mathbb{E}_{\pi_d}\left(  r_{(x,y,j,m),i}^{(\beta,\beta^{'})}  \Big| E_j^x,E_m^y\right) = -\frac{\epsilon ^2}{2}\left( I_j(\beta)+I_m(\beta)\right) 
	\label{eq:weightpresexp}
\end{equation}
and noting that there exists $|\xi_{I_m}|<\epsilon$ such that $I_m(\beta^{'})=I_m(\beta)+\epsilon J_m(\beta+\xi_{I_m})$
\begin{eqnarray}
	\mbox{Var}_{\pi_d}\left(  r_{(x,y,j,m),i}^{(\beta,\beta^{'})}\Big| E_j^x,E_m^y  \right) &=& \epsilon ^2\left(  \mbox{Var}_{f_{j}^{\beta}}\left[ \log f_{j}(X) \right] +\mbox{Var}_{f_{m}^{\beta^{'}}}\left[ \log f_{m }(X) \right]  \right) \nonumber\\
	&=&  \epsilon ^2 \left( I_j(\beta)+I_m(\beta)\right)+ \epsilon^3 J_m(\beta+\xi_{I_m}) . \label{eq:weightpresvar}
\end{eqnarray}

By  continuity of $J_k(\cdot)$, for all $k=1,\ldots,K$ then 
\[ \lim_{d\rightarrow \infty} J_k (\beta+\epsilon) = J_k (\beta) \]
and so there exists a bounding constant $C\in \mathbb{R}$ such that for all $\xi \in \{\xi_{T_1},\ldots,\xi_{T_K},\xi_{I_1},\ldots,\xi_{I_K}  \}$ and $k=1,\ldots,K$, \[ |J_k (\beta+\xi)| <C.\]
Consequently,
\[ \left|\epsilon^3 \sum_{i=1}^d J_k(\beta+\xi_{T_k})\right| \le \frac{\ell^3}{d^{1/2}}C \rightarrow 0 ~~\mbox{as}~~  d\rightarrow \infty. \]

With $\epsilon= \ell/d^{1/2}$, then using  (\ref{eq:weightpresexp}), (\ref{eq:weightpresvar}) and the central limit theorem for iid random variables, e.g. \cite{Durrett2010}, then as $d\rightarrow \infty$
\begin{equation}
	\sum_{i=1}^d r_{(x,y,j,m),i}^{(\beta,\beta^{'})} \Bigg| E_j^x,E_m^y \Rightarrow N\left( -\frac{\ell ^2}{2}\left( I_j(\beta)+I_m(\beta)\right),  \ell^2 \left( I_j(\beta)+I_m(\beta)\right) \right).  \nonumber
\end{equation}
Thus, by trivial use of Slutsky's Theorem 
\begin{equation}
B|E_j^x,E_m^y \Rightarrow N\left( -\frac{\ell ^2}{2}\left( I_j(\beta)+I_m(\beta)\right),  \ell ^2 \left( I_j(\beta)+I_m(\beta)\right) \right).  \nonumber
\end{equation}
and so with $w_k = \mathbb{P}_{\pi}(x \in A_k)$
\begin{equation}
	B \Rightarrow \sum_{j=1}^K\sum_{m=1}^K w_j w_m N\left( -\frac{\ell ^2}{2}\left( I_j(\beta)+I_m(\beta)\right),  \ell ^2 \left( I_j(\beta)+I_m(\beta)\right) \right) \label{eq:cltfull}
\end{equation}
which concludes the proof of Lemma~\ref{lemma:asymWB}.
\end{proof}

\subsection{Proof of Lemma~\ref{asymlemmaesjd}} \label{app:pf2}
\begin{proof}
Suppose that  $G\sim N(-\frac{\sigma^2}{2}, \sigma^2)$ then as derived in \cite{roberts1997weak},
\begin{equation}
\mathbb{E}\left(1\wedge e^G \right)= 2 \Phi\left(-\frac{\sigma}{2}\right) .\nonumber
\end{equation}
Using the definition of $ESJD_\beta$ with $\epsilon= \ell/d^{1/2}$ then
\[
d \times ESJD_\beta =\ell^2 \times  \mathbb{E}_{\pi_d}\left[ 1 \wedge e^B \right].
\]
Using the result of Lemma~\ref{lemma:asymWB} showing thet $B$ is a Gaussian mixture then
\begin{equation}
\lim_{d\rightarrow \infty}  \mathbb{E}_{\pi_d}\left[ 1 \wedge e^B \right] =\ell^2 \sum_{j=1}^K\sum_{m=1}^K w_j w_m \left(  2 \Phi\left(-\frac{\ell\sigma_{j,m}(\beta)}{2}\right)  \right), \nonumber
\end{equation}
where
$\sigma_{j,m}(\beta)= \left( I_j(\beta)+I_m(\beta)\right)^{\frac{1}{2}}$. Thus proving the result of Lemma~\ref{asymlemmaesjd}.

\end{proof}

\subsection{Proof of Lemma~\ref{accstepw}} \label{app:pf3}

\begin{proof}
This proof extends an argument of \cite{Sherlock2006}, where an optimal scaling result is  proven for spherically symmetric target distributions. Note also that herein $\phi(\cdot)$ denotes the density function of the standard Gaussian distribution.

First note that  the RHS of (\ref{eq:ESJDfinal}) can be expressed as
\begin{equation}
	\ell^2\mathbb{E}_{\Sigma_\beta} \left(    2 \Phi \left(-\frac{\ell\Sigma_{\beta}}{2}\right)  \right) \label{eq:ESJDRVform}
\end{equation}
where $\Sigma_{\beta}$ is a discrete RV such that $\mathbb{P}(\Sigma_{\beta}=\sigma_{i,j}(\beta))=w_iw_j$.

Differentiating \eqref{eq:ESJDRVform} with respect to $\ell$; setting equal to zero in order to find the optimal spacing, $\hat{\ell}$ and then rearranging gives
\begin{equation}
 2 \mathbb{E}_{\Sigma_{\beta}} \left(   \Phi\left(-\frac{\hat{\ell}\Sigma_{\beta}}{2}\right) \right)  = \mathbb{E}_{\Sigma_{\beta}}   \left(\frac{\hat{\ell}^2 \Sigma_{\beta}}{2} \phi\left(-\frac{\hat{\ell}\Sigma_{\beta}}{2}\right)  \right) . \label{eq:optidentityneed}
\end{equation}
As in \cite{Sherlock2006}, define the function $h(\cdot)$ defined so that
\begin{equation}
	h(x)= -\Phi^{-1}\left(x\right) \phi\left( \Phi^{-1}\left(x\right)\right)  \label{eq:hconcavefun}
\end{equation}
and note that $h(\cdot)$ is a concave function since for any $x \in (0,1)$ 
\begin{equation}
	\frac{\partial^2 h}{\partial x^2} =  -2 \frac{\Phi^{-1}\left(x\right)}{ \phi\left( \Phi^{-1}\left(x\right)\right)} < 0. \nonumber
\end{equation}

By considering the form of the $ESJD_\beta$ given in (\ref{eq:ESJDregional}) then it is clear that for any spacing  optimal acceptance rate, $\hat{a}$ is given by
\begin{equation}
	\hat{a} = \mathbb{E}_{\Sigma_{\beta}} \left(    2 \Phi\left(-\frac{\hat{\ell}\Sigma_{\beta}}{2}\right)  \right). \nonumber
\end{equation}
Letting $V:= \Phi\left(-\frac{\hat{\ell}\Sigma_{\beta}}{2}\right) $ then by (\ref{eq:optidentityneed}) at the optimal spacing 
\begin{equation}
	\hat{a} = 2 \mathbb{E}_{\Sigma_{\beta}} \left(    V  \right) = \mathbb{E}_{\Sigma_{\beta}} \left(  h(V)  \right) \nonumber
\end{equation}
where $h(\cdot)$ is given above in (\ref{eq:hconcavefun}). Since $h(\cdot)$ is concave then Jensen's inequality can be applied to give
$\mathbb{E}_{\Sigma_{\beta}} \left(  h(V)  \right) \leq  h \left(\mathbb{E}_{\Sigma_{\beta}} \left( V  \right) \right)$ 
and thus
\begin{equation}
	\hat{a} =  2\mathbb{E}_{\Sigma_{\beta}} \left(    V  \right) \leq -\Phi^{-1}\left(\mathbb{E}_{\Sigma_{\beta}} \left( V  \right)\right) \phi\left( \Phi^{-1}\left(\mathbb{E}_{\Sigma_{\beta}} \left( V  \right)\right)\right). \label{eq:keyineq2}
\end{equation}
Letting $m:= -\Phi^{-1}\left(\mathbb{E}_{\Sigma_{\beta}} \left( V  \right)\right)$ then by (\ref{eq:keyineq2})
\begin{equation}
	2\Phi\left(-m\right) \leq m \phi \left( -m \right), \label{eq:nearlythere}
\end{equation}
with equality only in the case when $m$ is the optimiser of the function $m^2\Phi\left(-m\right) $. Let this optimal $m$ be denoted $\hat{m}$ then 
\begin{equation}
	2\Phi\left(-\hat{m}\right) =  \hat{m}\phi\left( -\hat{m} \right) = 0.234 ~~\mbox{(3 s.f.)}.  \label{eq:rough0.234}
\end{equation}
\begin{figure}[h]
\begin{center}
\includegraphics[keepaspectratio,width=8cm]{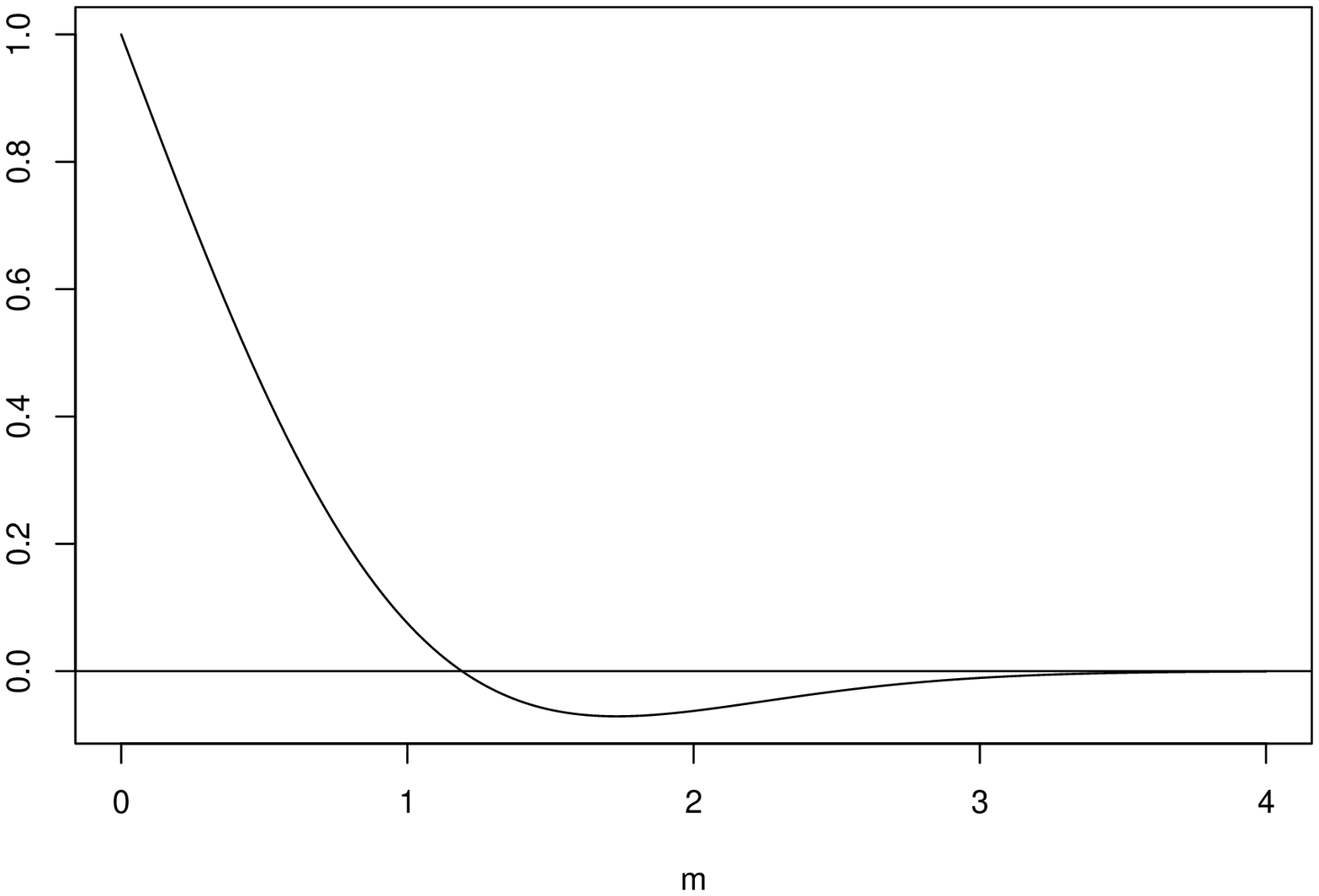}
\caption{Plot of the function 	$2\Phi\left(-m\right) - m \phi \left( -m \right)$}
\label{fig:asymplotscal}
\end{center}
\end{figure}

Figure~\ref{fig:asymplotscal} shows a plot of the function 	$2\Phi\left(-m\right) - m \phi\left( -m \right)$. Although not entirely clear in the figure, there is only  one point $\hat{m}$ giving a root of the function (occuring in the interval of $[0,2]$) and then for $m>\hat{m}$ the inequality given above in (\ref{eq:nearlythere}) holds strictly. Consequently, given that (\ref{eq:nearlythere}) holds then it implies that in this case $m>\hat{m}$ and so crucially due to the decreasing monotonicity of $\Phi(-m)$ the acceptance rate of the algorithm satisfies
\begin{equation}
	\hat{a}=2\Phi\left(-m\right) \leq 2\Phi\left(-\hat{m}\right) = 0.234 ~~\mbox{(3 s.f.)}.
\end{equation}

\end{proof}

\bibliographystyle{apalike}   

\bibliography{biblisim}

\begin{thebibliography}{}

\bibitem[Atchad{\'e} et~al., 2011]{atchade2011towards}
Atchad{\'e}, Y.~F., Roberts, G.~O., and Rosenthal, J.~S. (2011).
\newblock {Towards Optimal Scaling of {M}etropolis-Coupled {M}arkov chain
  {M}onte {C}arlo}.
\newblock {\em {Statistics and Computing}}, 21(4):555--568.

\bibitem[Bhatnagar and Randall, 2016]{Bhatnagar2016}
Bhatnagar, N. and Randall, D. (2016).
\newblock Simulated tempering and swapping on mean-field models.
\newblock {\em Journal of Statistical Physics}, 164(3):495--530.

\bibitem[Durrett, 2010]{Durrett2010}
Durrett, R. (2010).
\newblock {\em {Probability: Theory and Examples}}.
\newblock Cambridge university press.

\bibitem[Ge et~al., 2017]{Ge2017}
Ge, R., Lee, H., and Risteski, A. (2017).
\newblock Beyond log-concavity: Provable guarantees for sampling multi-modal
  distributions using simulated tempering langevin monte carlo.

\bibitem[Geyer, 1991]{geyer1991markov}
Geyer, C.~J. (1991).
\newblock {Markov chain {M}onte {C}arlo Maximum Likelihood}.
\newblock {\em {Computing Science and Statistics}}, 23:156--163.

\bibitem[Kou et~al., 2006]{kou2006discussion}
Kou, S., Zhou, Q., and Wong, W.~H. (2006).
\newblock {Equi-energy Sampler with Applications in Statistical Inference and
  Statistical Mechanics}.
\newblock {\em {The Annals of Statistics}}, pages 1581--1619.

\bibitem[Mangoubi et~al., 2018]{Mangoubi2018}
Mangoubi, O., Pillai, N.~S., and Smith, A. (2018).
\newblock Does hamiltonian monte carlo mix faster than a random walk on
  multimodal densities?

\bibitem[Marinari and Parisi, 1992]{marinari1992simulated}
Marinari, E. and Parisi, G. (1992).
\newblock {Simulated Tempering: a New {M}onte {C}arlo Scheme}.
\newblock {\em EPL (Europhysics Letters)}, 19(6):451.

\bibitem[Neal, 1996]{Neal1996}
Neal, R.~M. (1996).
\newblock {Sampling from Multimodal Distributions using Tempered Transitions}.
\newblock {\em {Statistics and Computing}}, 6(4):353--366.

\bibitem[{Nemeth} et~al., 2017]{2017arXiv170805239N}
{Nemeth}, C., {Lindsten}, F., {Filippone}, M., and {Hensman}, J. (2017).
\newblock {Pseudo-extended Markov Chain Monte Carlo}.
\newblock {\em ArXiv e-prints}.

\bibitem[Roberts et~al., 1997]{roberts1997weak}
Roberts, G.~O., Gelman, A., Gilks, W.~R., et~al. (1997).
\newblock {Weak Convergence and Optimal Scaling of Random Walk {M}etropolis
  Algorithms}.
\newblock {\em {The Annals of Applied Probability}}, 7(1):110--120.

\bibitem[Roberts and Rosenthal, 2014]{roberts2014minimising}
Roberts, G.~O. and Rosenthal, J.~S. (2014).
\newblock {Minimising {MCMC} Variance via Diffusion limits, with an Application
  to Simulated Tempering}.
\newblock {\em {The Annals of Applied Probability}}, 24(1):131--149.

\bibitem[Roberts et~al., 2001]{roberts2001optimal}
Roberts, G.~O., Rosenthal, J.~S., et~al. (2001).
\newblock {Optimal Scaling for Various Metropolis-Hastings Algorithms}.
\newblock {\em {Statistical Science}}, 16(4):351--367.

\bibitem[Sherlock, 2006]{Sherlock2006}
Sherlock, C. (2006).
\newblock {\em {Methodology for Inference on the Markov Modulated Poisson
  Process and Theory for Optimal Scaling of the Random Walk Metropolis}}.
\newblock PhD thesis, Lancaster University.

\bibitem[Tawn, 2017]{NTawnThesis}
Tawn, N. (2017).
\newblock {\em Towards {O}ptimality of the {P}arallel {T}empering {A}lgorithm}.
\newblock PhD thesis, University of Warwick.

\bibitem[Tawn and Roberts, 2018]{Tawn2018}
Tawn, N.~G. and Roberts, G.~O. (2018).
\newblock Accelerating parallel tempering: Quantile tempering algorithm
  (quanta).

\bibitem[Tawn et~al., 2018]{GarethJeffNick}
Tawn, N.~G., Roberts, G.~O., and Rosenthal, J.~S. (2018).
\newblock Weight preserving simulated tempering.

\bibitem[Wang and Swendsen, 1990]{Wang1990a}
Wang, J.-S. and Swendsen, R.~H. (1990).
\newblock Cluster monte carlo algorithms.
\newblock {\em Physica A: Statistical Mechanics and its Applications},
  167(3):565--579.

\bibitem[Woodard et~al., 2009a]{woodard2009conditions}
Woodard, D.~B., Schmidler, S.~C., and Huber, M. (2009a).
\newblock {Conditions for Rapid Mixing of Parallel and Simulated Tempering on
  Multimodal Distributions}.
\newblock {\em {The Annals of Applied Probability}}, pages 617--640.

\bibitem[Woodard et~al., 2009b]{woodard2009sufficient}
Woodard, D.~B., Schmidler, S.~C., and Huber, M. (2009b).
\newblock {Sufficient Conditions for Torpid Mixing of Parallel and Simulated
  Tempering}.
\newblock {\em {Electronic Journal of Probability}}, 14:780--804.

\end{thebibliography}
\end{document}